\documentclass[a4paper,twoside]{article}

\usepackage{epsfig}
\usepackage{subcaption}
\usepackage{calc}
\usepackage{amssymb}
\usepackage{amstext}
\usepackage{amsmath}
\usepackage{amsthm}
\usepackage{amsthm}
\usepackage{multicol}
\usepackage{pslatex}
\usepackage{apalike}
\usepackage{SCITEPRESS}     

\usepackage{listings}   

\usepackage{url}
\usepackage{cite}

\newtheorem{claim}{Claim}
\newtheorem{lemma}{Lemma}
\newtheorem{theorem}{Theorem}
\newtheorem{remark}{Remark}
\newtheorem{corollary}{Corollary}

\usepackage{hyperref}
\usepackage[nameinlink]{cleveref}

\crefname{table}{Table}{Tables}
\crefname{claim}{Claim}{Claims}

\newcommand{\KL}{\mathbf{KL}}
\newcommand{\Hent}{\mathbf{H}}
\newcommand{\JS}{\mathbf{JS}}
\newcommand{\twelch}{\mathbf{t}_{\mathrm{Welch}}}

\begin{document}

\title{Bounds on Bayes Factors for Binomial A/B Testing}

\author{\authorname{Maciej Skorski
\sup{1}
}
\affiliation{\sup{1}DELL}
\email{maciej.skorski@gmail.com}
}


\abstract{
Bayes factors, in many cases, have been proven to bridge the classic -value based significance testing and bayesian analysis of posterior odds. This paper discusses this phenomena within the binomial A/B testing setup (applicable for example to conversion testing). 
It is shown that the bayes factor is controlled by the \emph{Jensen-Shannon divergence} of success ratios in two tested groups, which can be further bounded by the Welch statistic. As a result, bayesian sample bounds almost match frequentionist's sample bounds.
The link between Jensen-Shannon divergence and Welch's test as well as the derivation are an elegant application of tools from information geometry.
}

\onecolumn \maketitle \normalsize \setcounter{footnote}{0} \vfill

\section{Introduction}
\label{sec:introduction}

\subsection{Motivation}

\paragraph{A/B testing}
A/B testing is the technique of collecting data from two parallel experiments and comparing them by probabilistic inference.
A particularly important case is assessing which of two success-counting experiments achieves a higher success rate. This naturally applies to evaluating conversion rates on two different versions of a webpage. A typical question being asked is
if there is a difference (called also non-zero effect) in conversion between groups: $p_1,p_2$ are unknown conversion rates in two experiments, and the task is to compare hypotheses 
$H_0 = \{p_1=p_2\}$ and $H_a = \{p_1\not=p_2\}$, given observed data. In the frequentionist approach, one falsifies $H_0$ by the \emph{two-sample t-test}\cite{Welch}. In the bayesian approach one \emph{evaluates the strength of both} $H_0$ and $H_a$ and decides based the ratio called \emph{bayes factor}
\begin{align}
K=\frac{\Pr[  \mathcal{D} | H_{0}] }{\Pr[  \mathcal{D} | H_{a}] }
\end{align}
which converted by the Bayes theorem to $\frac{\Pr[ H_0 |  \mathcal{D} ] }{\Pr[  H_{a} |  \mathcal{D} ] } = K\cdot \frac{\Pr[H_0]}{\Pr[H_a]}$ quantifies \emph{posterior odds} and allows a research to choose a model \emph{more plausible given data} (usually one gives $H_0$ and $H_a$ same chance of getting considered and sets $\Pr[H_0]=\Pr[H_a]=\frac{1}{2}$). The decision rule and confidence depends on the magnitude of $K$~\cite{Kass95,jeffreys1998theory}. In the bayesian approach a hypothesis assigns an arbitrary distribution to
parameters which is more general.

\paragraph{Testing counts proportions}
Suppose that empirical data $\mathcal{D}$ has $r_i=r$ runs and $r\cdot \bar{p}_i$ successes in the $i$-th experiment, $i=1,2$. Under the binomial counting model, the data likelihood under a hypothesis $H$ equals
\begin{align}
\Pr[\mathcal{D}|H]=\int \prod_{i=1}^{2} p_i^{\bar{p}_i r}(1-p_i)^{(1-\bar{p})_i r}\mbox{d}\mathbb{P}_{H}(p_1,p_2)
\end{align}
where the \emph{prior} distribution $\mathbb{P}_{H}(\cdot,\cdot)$ reflects what is assumed prior to seeing data (and what will be tested); one can for example choose $\{p_1=p_2=0.1\%\}$  for $H=H_0$ and $\{p_1\not=p_2\}$ for $H_a$ uniformly over all valid values of $p_1,p_2$, but in practice more informative priors are used because some configurations of values are unrealistic (e.g. extremely low or high conversion). The corresponding factor $K$ can be computed for example by the R package \texttt{BayesFactor}~\cite{bayesfactor2018}.

\paragraph{Problem: Bayesian A/B testing power}
Estimates, neither frequentist nor bayesian, will not be conclusive without sufficiently many samples.
Frequentists widely use rules of thumbs that are derived based on t-tests. Under the bayesian methodology this is little more complicated because hypotheses can be arbitrary priors over parameters. Under the binomial A/B model, we will answer the following questions
\begin{itemize}
\item when, given data, a \emph{bayesian} hypothesis on zero effect may be rejected ($K \ll 1$ for some $K_a$)?
\item what is the relation to the classical t-test? 
\end{itemize}
This will allow us to understand \emph{data limitations} when doing bayesian inference, and relate
them to widely-spread frequentionist rule of thumbs.

\subsection{Related Works and Contribution}

Our problem, as stated, is a question about \emph{maximizing minimal bayes factor}. It is known that for certain problems bayes factors can be related to frequentionist's p-values~\cite{Edwards1963,Kass95,Goodman1999TowardEM} and thus bridges the Bayesian and frequentionist world (this should be contrasted with a wide-spread belief that both methods are very incompatible~\cite{Kruschke2018}). The novel  contributions of this paper are (a) bounding the Bayes factor for binomial distributions (b) discussion of sample bounds for binomial A/B testing in relation to the frequentionist approach.

\paragraph{Main result: Bayes factor and Welch's statistic}

The following theorem shows that no ``zero-effect'' hypothesis can be falsified, unless the number of samples is big in relation to a certain \emph{dataseet statistic}. This statistic turns out to be the Jensen-Shannon divergence, well-known in information theory. It is in turn bounded by the Welch's t-statistic.
\begin{theorem}[Bayes Factors for Binomial Testing]\label{thm:main}
Consider two independent experiments, each with $r$ independent trials with unknown success probabilities $p_1$ and $p_2$ respectively. Let observed data $\mathcal{D}$ has $r\cdot \bar{p}_i$ successes and $r\cdot (1-\bar{p}_i)$ failures for group $i$. Then
\begin{align}
\max_{H_0: \{p_1=p_2\} } \min_{H_a} \frac{\Pr[H_0|\mathcal{D}]}{\Pr[H_a|\mathcal{D}]} = \mathrm{e}^{-2r\cdot \JS(\bar{p_1},\bar{p_2})} 
\end{align}
where the maximum is over null hypothesis (priors) $H_0$ over $p_1,p_2$ such that $p_1=p_2$, the minimum is over
all valid alternative hypothesis (priors) over $p_1,p_2$, and $\JS$ denotes the Jensen-Shannon divergence. 

Moreover, the Jensen-Shannnon divergence is bounded by the Welch's $t$-statistic (on $\mathcal{D}$) 
\begin{align}
\JS(\bar{p_1},\bar{p_2}) \geqslant\frac{ \twelch(t,\bar{p}_1,\bar{p}_2)^2}{4r}
\end{align}
so that we can bound
\begin{align}
\max_{H_0: \{p_1=p_2\} } \min_{H_a} \frac{\Pr[H_0|\mathcal{D}]}{\Pr[H_a|\mathcal{D}]} \leqslant \mathrm{e}^{-\twelch(t,\bar{p}_1,\bar{p})^2/2} 
\end{align}
\end{theorem}

\begin{remark}[Most favorable hypotheses] Note that
\begin{itemize}
\item Maximally favorable alternative ($H_a$ which maximizes $\Pr[\mathcal{D}|H_a]$) is $p_1 = \theta_1$ and $q_1=\theta_2$
\item Maximally favorable null of the form $p_1=q_1$ is $p_1=q_1=\frac{\bar{p_1}+\bar{p_2}}{2}$
\end{itemize}
If null is of the form $p=q=x_0$ then the bound becomes $\mathrm{e}^{-r\cdot \KL(\bar{p}_1,x_0)-r\cdot \KL(\bar{p}_2,x_0)}$.
\end{remark}
\paragraph{Application: sample bounds}
The main result implies the following sample rule
\begin{corollary}[Bayesian Sample Bound]\label{cor:sample_bound}
To confirm the non-zero effect ($p_1\not=p_2$) the number of samples for the bayesian method should be
\begin{align}\label{eq:1}
r \gg \frac{1}{2\JS(p_1,p_2)}
\end{align}
\end{corollary}
Under the frequenionist method the rule of thumb is $\twelch \gg 1$, which gives (see \Cref{sec:prelim})
\begin{align}\label{eq:2}
r \gg  \frac{2(p_1(1-p_1) +p_2(1-p_2)}{(p_1-p_2)^2}
\end{align}
Note that both formulas needs assumptions on locations of the parameters. In particular, testing smaller effects or effects with higher variance require more samples.

Bounds \Cref{eq:1} and \Cref{eq:2} are close to each other by a constant factor (a different small factor is necessary to make the bound small in both the bayesian credibility and p-value sense). The difference (under the normalized constant) is illustrated on \Cref{fig:example}, for the case when one wants to test a relative uplift of $10\%$.

\begin{figure}
\includegraphics[width=0.49\textwidth]{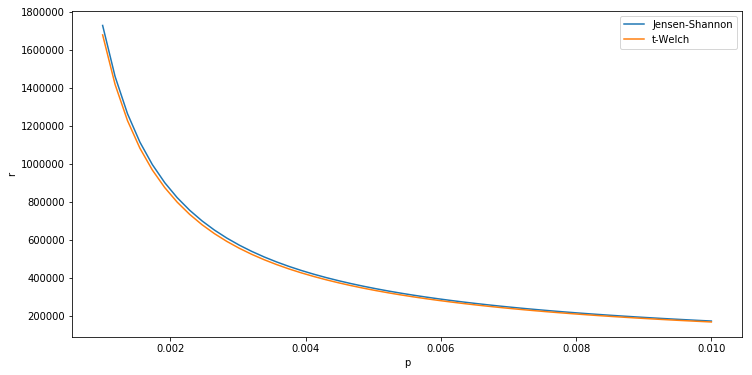}
\caption{Comparison of the bayesian~\eqref{eq:1} (bayesian) and the frequentionist~\eqref{eq:2} sample lower bounds, where $p_2=p$ and $p_1=p_1\cdot (1+\delta)$ for $\delta=0.1$ (10\% uplift). Both formulas are multiplied by a factor of 2 to accommodate meaningful confidence.}
\label{fig:example}
\end{figure}

Since high values of $\twelch$ means small p-values, we conclude that the frequentionist p-values bounds the bayes factor and indeed, are evidence against a null-hypothesis in the well-defined bayesian sense. However, because of the scaling
$\twelch \to \mathrm{e}^{-\twelch^2/2}$, this is true for p-values much lower than the standard threshold of 0.05.
In some sense, the bayesian approach is more conservative and less reluctant to reject than frequentionist tests; this conclusion is shared with other works~\cite{Goodman1999TowardEM}.


\section{Preliminaries}

\paragraph{Entropy, Divergence}

The binary cross-entropy of $p$ and $q$ is defined by
\begin{align}
\Hent(p,q) = -p\log (1-p) - (1-p)\log(1-q)
\end{align}
which becomes the standard (Shannon) binary entropy when $p=q$, denoted as $\Hent(p) = \Hent(p,p)$. The Kullback-Leibler divergence is defined as
\begin{align}
\KL(p,q) = \Hent(p,q)  - \Hent(p)
\end{align}
and the Jensen-Shannon divergence~\cite{61115} is defined as
\begin{align}
\JS(p,q) = \Hent(p,q)  - \frac{1}{2}\Hent(p)-\frac{1}{2}\Hent(q)
\end{align}
(always positive because the entropy is concave).

The following lemma shows that the cross-entropy function is \emph{convex} in the second argument. This should be contrasted with the fact that the entropy function (of one argument) is concave.
\begin{lemma}[Convexity of cross-entropy]\label{lemma:cross_ent_convex}
For any $p$ the mapping $x\to \Hent(p,x)$ is convex in $x$.
\end{lemma}
\begin{proof}
Since $-p\cdot \log(\cdot)$ for $p\in [0,1]$ is convex we obtain 
\begin{align*}
-\gamma_1 p \log x_1 - \gamma_2p  \log x_2 \geqslant -p\log (\gamma_1 x_1 + \gamma_2 x_2)
\end{align*}
for any $x_1,x_2$ and any $\gamma_1,\gamma_2 \geqslant 0$, $\gamma_1+\gamma_2 = 1$.
Replacing $x_i$ by $1-x_i$ and $p$ by $1-p$ in the above inequality gives us also
\begin{multline*}
-\gamma_1 (1-p) \log (1-x_1) - \gamma_2 (1-p)  \log (1-x_2) \\ \geqslant -(1-p)\log (\gamma_1 (1-x_1) + \gamma_2 (1-x_2)) \\
= -(1-p)\log (1-\gamma_1x_1) - \gamma_2 x_2)
\end{multline*}
Adding side by side yields
\begin{align*}
\gamma_1\Hent(p,x_1)+\gamma_2\Hent(p,x_2) \geqslant \gamma_1\Hent(p,x_1)+\gamma_2\Hent(p,x_2)
\end{align*}
which finishes the proof.
This argument works for multivariate case, when $p,x$ are probability vectors.
\end{proof}

\begin{lemma}[Quadratic bounds on KL/cross-entropy]\label{thm:ent_quadratic}
For any $p$ it holds that
\begin{align}
\KL(p,x) \geqslant \left(\frac{1}{p}+\frac{1}{1-p}\right)\cdot (x-p)^2
\end{align}
\end{lemma}
\begin{proof}
We will prove a general version. Let $(p_i)_i$ and $(x_i)_i$ be probability vectors of the same length.
By the elementary inequality
\begin{align}
\log (1+u)\geqslant u-\frac{1}{2}u^2
\end{align}
we obtain
\begin{align}
-\log (x_i/p_i) &=- \log (1-(p_i-x_i)/p_i) \geqslant \\
&-\frac{p_i-x_i}{p_i} +\frac{1}{2}\left(\frac{p_i-x_i}{p_i}\right)^2
\end{align}
multiplying both sides by $p_i$ and adding inequalities side by side we obtain
\begin{align}
-\sum p_i\log (x_i/p_i) &  \geqslant -\sum_i (x_i-p_i) + \sum_{i}\frac{\left(p_i-x_i\right)^2}{2p_i}\\
& =\sum_{i}\frac{\left(p_i-x_i\right)^2}{2p_i}
\end{align}
which means $\KL(x,p)\sum_{i}\frac{\left(p_i-x_i\right)^2}{2p_i}$. Our lemma follows by specializing to the vectors $(p,1-p)$ and $(x,1-x)$.
\end{proof}

\paragraph{2-Sample test}\label{sec:prelim}

To decide whether means in two groups are equal, under the assumption of unequal variances, one performs the Welch's t-test with the statistic
\begin{align}
\twelch = \frac{\mu_1-\mu_2}{\sqrt{\frac{s_1^2}{r_1} + \frac{s_2^2}{r_2}}}
\end{align}
where $s_i$ are sample variances and $\mu_i$ are sample means for group $i=1,2$. 
The null hypothesis is rejected unless the statistic is sufficiently high (in absolute terms). In our case the formula simplifies to

\begin{claim}\label{claim:welch}
If $r \theta_1$ and $r \theta_2$ success out of $r$ trials have been observed respectively in the first and the second group then
\begin{align}
\twelch(t,\theta_1,\theta_2) = r^{-\frac{1}{2}}\cdot \frac{\theta_1-\theta_2}{\sqrt{\theta_1(1-\theta_1) + \theta_2(1-\theta_2)}}
\end{align}
\end{claim}

\section{Proof}

We change the notation slightly, unknown success rates will be $p$ and $q$, and corresponding successes $r\cdot \theta_1,r\cdot\theta_2$.
\paragraph{Alternatives}
Maximizng over \emph{alll} posible priors $\mathbf{P}_a$ over pairs $(p,q)$ we get
\begin{multline}\label{eq:like_a}
\max_{\mathbb{P}}\Pr[  \mathcal{D} | H_{a}] = \\ c\cdot\max_{\mathbb{P}}\int_{[0,1]^2} \mathrm{e}^{-r H(\theta_1,p)-r H(\theta_1,q)}\mathbb{P}_a(p,q)\mbox{d}(p,q) 
\end{multline}
where $c = \frac{1}{B(r\theta_1+1,r(1-\theta_1)+1)\cdot B(r\theta_2+1,r(1-\theta_2)+1)}$ is a normalizing constant,
which equals
\begin{align}
\max_{\mathbb{P}_a}\Pr[  \mathcal{D} | H_{a}] = c \cdot \mathrm{e}^{-r H(\theta_1) - rH(\theta_2)}
\end{align}
achieved for $\mathbf{P}_a$ being a unit mass at $(p,q)=(\theta_1,\theta_2)$.

\paragraph{Null}
Let $H_0$ states that the baseline is $p$ and the efect is $0$. Then we obain
\begin{align}\label{eq:like_0}
\Pr[  \mathcal{D} | H_{0}]  = c\cdot \mathrm{e}^{-rH(\theta_1,p)-rH(\theta_2,p)}
\end{align}
with the same normalizing constant $c$.

\paragraph{Bayes factor}
If none of two hypothesis is a priori prefered, that is when $\Pr[H_{0}]=\Pr[H_a]$, then the Bayes factor equals
the likelihood ratio (by Bayes theorem)
\begin{align}
\frac{\Pr[  H_{0}|\mathcal{D} ]}{\Pr[  H_{a} |  \mathcal{D}]} = \frac{\Pr[  \mathcal{D} | H_{0}]}{\Pr[  \mathcal{D} | H_{a}]}.
\end{align}
In turn the likelihood ratio (in favor of $H_0$) equals
\begin{align}
\min_{H_a}\frac{\Pr[  \mathcal{D} | H_{0}]}{\Pr[  \mathcal{D} | H_{a}]} = \mathrm{e}^{-r\cdot\left( \Hent(\theta_1,p) + \Hent(\theta_1,p) - \Hent(\theta_1)-\Hent(\theta_2) \right)}
\end{align}
(the normalizing constant $c$ cancells). Using the relation between the KL divergence and cross-entropy we obtain
\begin{align}\label{eq:b_factor1}
\min_{H_a}\frac{\Pr[  \mathcal{D} | H_{0}]}{\Pr[  \mathcal{D} | H_{a}]} = \mathrm{e}^{-r \KL(\theta_1,p)-r \KL(\theta_2,p)}
\end{align}
We will use the following observation
\begin{claim}
The expression $\KL(\theta_1,p)+\KL(\theta_2,p)$ is minimized under $p=\theta^{*}=\frac{\theta_1+\theta_2}{2}$, and achieves value $2\JS(\theta_1,\theta_2)$.
\end{claim}
\begin{proof}
We have
\begin{multline*}
\KL(\theta_1,p)+\KL(\theta_2,p) \\= \Hent(\theta_1,p)+\Hent(\theta_2,p) - \Hent(\theta_1)-\Hent(\theta_2) 
\end{multline*}
Now the existence of the minimum at $p=\theta^{*}$ follows by convexity of $p\to \Hent(\theta_1,p)+\Hent(\theta_2,p) $, proved in \Cref{lemma:cross_ent_convex}.
We note that $\Hent(\theta_1,p)+\Hent(\theta_2,p)  = 2\Hent\left(\frac{\theta_1+\theta_2}{2},p\right)$ for any $p$ (by definition), and thus for $p=\frac{\theta_1+\theta_2}{2}=\theta^*$ we obtain $\Hent(\theta_1,p)+\Hent(\theta_2,p) =2H(\theta^*)$ and $\KL(\theta_1,p)+\KL(\theta_2,p) = 2\Hent(\theta^*)-\Hent(\theta_1)-\Hent(\theta_2) $. This combined with the definition of the Jensen-Shannon divergence finishes the proof.
\end{proof}  
We can now bound \Cref{eq:b_factor1} as
\begin{align}\label{eq:b_factor2}
\min_{H_a}\frac{\Pr[  \mathcal{D} | H_{0}]}{\Pr[  \mathcal{D} | H_{a}]} \leqslant \mathrm{e}^{-2r\cdot \JS(\theta^{*})}
\end{align}
This proves the first part of \Cref{thm:main}
\paragraph{Connecting t-statistic and bayes factor exponent}
Recall that by \Cref{claim:welch} under t-test we have 
\begin{align}\label{eq:t_student}
T \approx r^{\frac{1}{2}}\cdot |\theta_1-\theta_2|\cdot \left(\theta_1(1-\theta_1) + \theta_2(1-\theta_2) \right)
\end{align}
It remains to connect $|\theta_1-\theta_2|$ and $\JS(\theta_1,\theta_2)$. By \Cref{thm:ent_quadratic} we have the following refinement of Pinsker's inequality
\begin{claim}\label{claim:kl_ineq}
We have $\KL(\theta,p)\geqslant \frac{(\theta-p)^2}{2\theta(1-\theta)}$. 
\end{claim}
Using $2\JS(\theta_1,\theta_2) = \KL(\theta_1,\theta^{*}) +  \KL(\theta_2,\theta^{*})$,
the inequality from \Cref{claim:kl_ineq}, and the Welch's formula in \Cref{eq:t_student} we obtain
\begin{claim}
We have
\begin{align}\label{eq:w_js}
\JS(\theta_1,\theta_2) \geqslant\frac{ \twelch(t,\theta_1,\theta_2)^2}{4r}
\end{align}
\end{claim}
\begin{proof}
\Cref{claim:kl_ineq} implies
\begin{align}
2\JS(\theta_1,\theta_2)\geqslant (\theta_1-\theta_2)^2\cdot \left(\frac{1}{2\theta_1(1-\theta_1)}+\frac{1}{2\theta_1(1-\theta_2)}\right)
\end{align}
we recognize the Weltch's statistic and write
\begin{align}
2\JS(\theta_1,\theta_2)\geqslant \frac{ \twelch(t,\theta_1,\theta_2)^2}{2r}
\end{align}
\end{proof}
Combining \Cref{eq:b_factor2} and \Cref{eq:w_js} implies the second part of the theorem.


\section*{Acknowledgments}
The author thanks to Evan Miller for inspiring discussions.

\bibliographystyle{apalike}
{\small
\bibliography{citations}}

\appendix

\end{document}